\documentclass{article}
\usepackage[utf8]{inputenc}
\usepackage{amsmath, amsthm, amsfonts,setspace,graphicx}
\usepackage[usenames, dvipsnames]{color}
\usepackage{fullpage}
\usepackage{dsfont}
\usepackage{natbib}
\usepackage{color}   
\usepackage{hyperref}
\hypersetup{
    colorlinks=true, 
    linktoc=all,     
    linkcolor=blue  
}

\doublespace

\newtheorem{definition}{Definition}

\newtheorem{theorem}{Theorem}

\newtheorem{proposition}{Proposition}

\newtheorem{example}{Example}
\newtheorem{remark}{Remark}

\newcommand{\A}{\mathcal{A}}
\newcommand{\Sall}{\mathcal{S}}

\newcommand{\Img}{\mathrm{Img}}

\newcommand{\Ccal}{\mathcal{C}}

\title{A general theory of identification}
\author{Guillaume Basse \\ Department of MS\&E and Department of Statistics\\ Stanford
\and Iavor Bojinov \\ Harvard Business School \\ Harvard University}

\begin{document}

\maketitle



\begin{abstract}
What does it mean to say that a quantity is identifiable from the data? Statisticians seem to agree on a definition in the context of parametric statistical models --- roughly, a parameter $\theta$ in a model $\mathcal{P} = \{P_\theta: \theta \in \Theta\}$ is identifiable if the mapping $\theta \mapsto P_\theta$ is injective. This definition raises important questions: Are parameters the only quantities that can be identified? Is the concept of identification meaningful outside of parametric statistics? Does it even require the notion of a statistical model? Partial and idiosyncratic answers to these questions have been discussed in econometrics, biological modeling, and in some subfields of statistics like causal inference. This paper proposes a unifying theory of identification that incorporates existing definitions for parametric and nonparametric models and formalizes the process of identification analysis. The applicability of this framework is illustrated through a series of examples and two extended case studies.

\end{abstract}


\section{Introduction}
Statistical inference teaches us ``how'' to learn from data, whereas identification analysis explains ``what'' we can learn from it. Although ``what'' logically precedes ``how,'' the concept of identification has received relatively less attention in the statistics community. In contrast, economists have been aware of the identification problem since at least the 30's \citep[][Chapter 9]{frisch1934statistical} and have pioneered most of the research on the topic. \cite{koopmans1949identification} coined the term ``identifiability'' and emphasized a ``clear separation between problems of statistical inference arising from the variability of finite samples, and problems of identification in which [the statistician] explore[s] the limits to which inference even from an infinite number of observations is subject.''

\cite{hurwicz1950generalization} and \cite{koopmans1950identification} formalized the intutive idea of identification and developed a general theory for statistical models. The literature then fractures, with specialized definitions of identifiability arising in different fields, including biological systems modeling \citep{jacquez1990parameter}; parametric models \citep{rothenberg1971identification,hsiao1983identification, paulino1994identifiability}; ecological regression \citep{goodman1959some,cross2002regression}; nonparametric models \citep{matzkin2007nonparametric}; causal models \citep{pearl2009causal,shpitser2008complete}; and nonparametric finite population models \citep{manski1989anatomy,manski2009identification}. This divergence and lack of coherent unifying theory has obfuscated some central ideas and slowed down the development of the field.

This paper proposes a general framework for studying identifiability that encompasses existing definitions as special cases. We make three main contributions. First, we study the common structure of the specialized definitions and extract a single general --- and mathematically rigorous --- definition of identifiability. Abstracting away the specifics of each domain allows us to recognize the commonalities and make the concepts more transparent as well as easier to extend to new settings.  Second, we use our definition to develop a set of results and a systematic approach for determining whether a quantity is identifiable and, if not, what is its identification region (\emph{i.e.,} the set of values of the quantity that are coherent with the data and assumptions). This process of {\em identification analysis}, formalizes ideas introduced in the literature on partial identification \citep{manski2003partial,manski2009identification, tamer2010partial}. Third, we provide concrete examples of how to apply our definition in different settings and include two in-depth case studies of identification analysis.

The paper proceeeds as follows. Section~\ref{section:theory} introduces our
general theory, starting with some backgound on binary relations
(Section~\ref{section:background}), which are the key mathematical objects
underpinning our definition of identification
(Section~\ref{section:identifiability}). We illustrate the flexibility and
broad applicability of our definition in Section~\ref{section:id-stat-models}
and discuss identification analysis in Section~\ref{section:id-analysis}.
Finally, we provide two case studies in Section~\ref{section:case-studies}.


\section{General theory of identification}
\label{section:theory}

\subsection{Background on binary relations}
\label{section:background}

Let $\Theta$ and $\Lambda$ be two sets. A binary relation $R$ from 
$\Theta$ to $\Lambda$ is a subset of the cartesian product 
$\Theta \times \Lambda$. For $\vartheta \in \Theta$ and $\ell \in \Lambda$, we say
that $\vartheta$ is $R$-related to $\ell$ if $(\vartheta, \ell) \in R$. Following 
convention \citep{halmos2017naive}, we use the notation $\vartheta R \ell$ as an abbreviation for
$(\vartheta, \ell) \in R$. Below, we define four important
properties that a binary relation may have \citep{freedman2015some}; see \cite{lehman2010mathematics} for an in-depth discussion.
\begin{definition} \label{D:binary-relation-properties}
  A binary relation $R$ from $\Theta$ to $\Lambda$ is said 
  to be:
  \begin{itemize}
  \item {\em injective} if 
    \begin{equation*}
      \forall \vartheta, \vartheta' \in \Theta, \forall \ell \in \Lambda,
      \qquad \vartheta R \ell \text{ and } \vartheta' R \ell 
      \quad \Rightarrow \quad 
      \vartheta = \vartheta'
    \end{equation*}
    
  \item {\em surjective} if 
    \begin{equation*}
      \forall \ell \in \Lambda, \exists \vartheta \in \Theta: 
      \quad \vartheta R \ell
    \end{equation*}
    
  \item {\em functional} if 
    \begin{equation*}
      \forall \vartheta \in \Theta, \ \forall \ell,\ell'\in \Lambda,
      \qquad \vartheta R \ell \text{ and } \vartheta R \ell' 
      \quad \Rightarrow \quad \ell = \ell'
    \end{equation*}
    
  \item {\em left-total} if 
    \begin{equation*}
      \forall \vartheta \in \Theta,\ \exists \ell \in \Lambda: 
      \quad \vartheta R \ell
    \end{equation*}
    
  \end{itemize}
\end{definition}
A binary relation that is both functional and left-total is called a function. 
\begin{example}
  Let $\Theta$ be the set of prime numbers, $\Lambda$ be the set of integers, and
  $R$ the ``divides'' relation such that $\vartheta R \ell$ if $\vartheta$ divides
  $\ell$ (e.g., $3 R 3$, $3 R 6$, but $3$ is not in relation with 2). In this case,
  $R$ is surjective and left-total, but not injective nor functional.
\end{example}
\begin{example} \label{example:square}
  Let $\Theta = \mathbb{R}$, $\Lambda = \mathbb{R}$, and $R$ be the ``square''
  relation defined by $\vartheta R \ell$ if $\vartheta^2 = \ell$. In this case, $R$ is
  left-total and functional, but it is not surjective (e.g., there is no
  $\vartheta \in \Theta$ such that $\vartheta R (-4)$) nor injective (e.g., $2R4$ and $-2R4$).
  If we instead consider $\Lambda = \mathbb{R}_{\ge 0}$, the set of all
  positive real numbers and 0, then $R$ is both surjective and injective. 
\end{example}

Example \ref{example:square}, shows that the properties described in Definition~\ref{D:binary-relation-properties} depend on both the binary relation and the sets $\Lambda$ and $\Theta$. Throughout this paper, whenever we refer to properties of binary relations, the dependence on $\Lambda$ and $\Theta$ will always be implied.

\subsection{Identification in sets and functions}
\label{section:identifiability}

We start by defining identifiability for a binary relation. The definition forms the basis of our unifying framework as all the other definitions of identifiability are obtainable by specifying appropriate $\Lambda$, $\Theta$, and $R$.

\begin{definition}[Identifiability]\label{D:identifiability-binary}
	Let $\Theta$ and $\Lambda$ be two sets, and $R$ a surjective and 
    left-total binary relation from $\Theta$ to $\Lambda$. Then,
    \begin{itemize}
    	\item $\Theta$ is {\em R-identifiable} at $\ell_0 
        \in \Lambda$ if there exists a unique $\vartheta_0 \in \Theta$ 
        such that $\vartheta_0 R \ell_0$;
        \item $\Theta$ is {\em everywhere R-identifiable} in 
        $\Lambda$ if it is {\em R-identifiable} at $\ell_0$ for all 
        $\ell_0 \in \Lambda$. In this case, we usually say that $\Theta$ is $R$-identifiable. 
    \end{itemize}
\end{definition}
The distinction between $R$-identifiable at $\ell_0$ and everywhere has important practical implications. For example, when handling missing data, the missing at random assumption aims to obtain identification at the observed missing data pattern (\emph{i.e.}, at $\ell_0$); whereas, the stronger missing always at random aims to be everywhere identifiable \citep{bimetrkia2015Mealli,bojinov2017diagnostics}.

Formally, identifiability everywhere is equivalent to the binary relation being injective.
\begin{proposition}\label{prop:injective}
	Let $\Theta$ and $\Lambda$ two sets, and $R$ a surjective and 
    left-total binary relation from $\Theta$ to $\Lambda$. 
    $\Theta$ is $R$-identifiable if and only if $R$ is injective.
\end{proposition}

\begin{proof}
This is a restatement of the definition. 
\end{proof}
\begin{figure}
	\centering
	\includegraphics[scale=0.5]{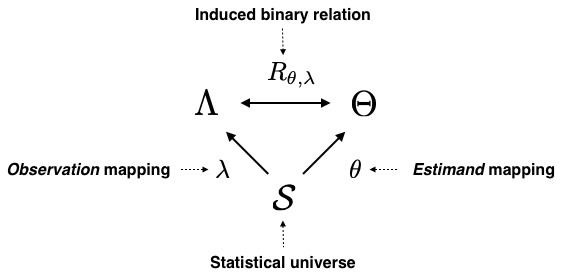}
    \caption{Diagram of the main objects.}
    \label{figure:diagram-basic}
\end{figure}

In most practical applications, we can derive a natural specification of the problem by working with an \emph{induced binary relation}. Intuitively, an induced binary relation connects ``what we know'' to ``what we are trying to learn'' through a ``statistical universe'' in which we operate. 
\begin{definition}[Induced binary relation]
\label{definition:ind-bin-rel}
Let $\Sall$ be a set and $G(\Sall)$ be the set of all functions with domain 
$\Sall$. Let $\lambda, \theta \in G(\Sall)$, and 
$\Theta = \Img(\theta)$ and $\Lambda = \Img(\lambda)$ their respective 
images. The binary relation from $\Theta$ to $\Lambda$ defined as 
$R_{\theta, \lambda} = \{(\theta(S), \lambda(S)), S\in \Sall\}$ is called 
the {\em induced binary relation} associated with $(\theta, \lambda)$.
\end{definition}
The examples in Section~\ref{section:id-stat-models} show how $\Sall$, $\lambda$
and $\theta$ map to real problems. In broad terms, the
{\em statistical universe} $\Sall$ contains all the objects relevant to a given
problem; the {\em observation mapping} $\lambda$ maps $\Sall$ to ``what we
know''; and the {\em estimand mapping} $\theta$ maps $\Sall$ to ``what we are
trying to learn''. Figure \ref{figure:diagram-basic},  illustrates how these concepts are connected.

The following proposition follows immediately from the definition of an induced binary relation.
\begin{proposition}
	Let $\Sall$ be a set, and $\theta,\lambda \in G(\Sall)$. The induced 
    binary relation $R_{\theta,\lambda}$ is surjective and left-total.
\end{proposition}
Applying Definition \ref{definition:ind-bin-rel} to the induced binary relation allows us to extend the notion of identification from sets to functions.

%
\begin{definition}[Identifiability of a function]\label{D:identifiability-function}
Consider $\Sall$ and $\theta, \lambda \in G(\Sall)$, and let 
$\Theta = \Img(\theta)$ and $\Lambda = \Img(\lambda)$.
\begin{itemize}
	\item The function $\theta$ is said to be identifiable at 
    $\ell_0 \in \Lambda$ if $\Theta$ is 
    $R_{\theta,\lambda}$-identifiable at $\ell_0$. That is, for $\ell_0 \in \Lambda$ let $\mathcal{S}_0 = \{S \in \mathcal{S}  : \lambda(S) = \ell_0 \}$ then $R_{\theta, \lambda}$ is identifiable at $\ell_0$ iff there exists $\vartheta_0 \in \Theta$, such that, for all $S \in \mathcal{S}_0$, we have that $\theta(S)=\vartheta_0$.
    \item The function $\theta$ is said to be identifiable everywhere 
    from $\lambda$ if $\Theta$ is $R_{\theta,\lambda}$-identifiable 
    everywhere in $\Lambda$. We will usually simply say that 
    $\theta$ is identifiable.
\end{itemize}
\end{definition}

Definition~\ref{D:identifiability-function} is the workhorse allowing us to unify
the notions of identifiability used in the literature for parametric and nonparametric models, as well as for finite populations.

\begin{remark}
Both Definitions~\ref{D:identifiability-binary} and \ref{D:identifiability-function} use the adjective ``identifiable'' to qualify a set $\Theta$ or a mapping $\theta$. The terminology arises naturally from the interpretation of $\theta$ and $\lambda$; indeed, we write that the estimand mapping $\theta$ is identifiable from the observation mapping $\lambda$. Proposition~1, however, makes it clear that identifiability is fundamentally a property of the binary relation $R$, whether we apply the adjective identifiable to $\theta$, $\lambda$, or the whole model is mostly a matter of semantics. 
\end{remark}

\section{Identification in statistical models and finite populations}
\label{section:id-stat-models}

There are two significant benefits of using our framework to tackle identification in statistical models. First, the flexibility of our general formulation allows us to work directly with both parametric and nonparametric models, without having to introduce separate definitions. Second, relying on binary relations instead of functions enriches the class of questions that can be addressed through
the lens of identification.

In this section, we make extensive use of examples to illustrate the broad applicability of our framework. All examples follow a common structure: first we explain the context; second, we ask an informal identification question; third, we show how to formalize the question in our framework by specifying $\Sall$, $\lambda$, and $\theta$ appropriately. The process of answering these questions, which we call {\em identification analysis} is described in Section~\ref{section:id-analysis} and illustrated in Section~\ref{section:case-studies}.

\subsection{Parametric models}

Consider a parametric model $\Lambda = \{P_\vartheta, \vartheta \in \Theta\}$,
where $\Theta$ is a finite dimensional parameter space and $P_\vartheta$ is a
distribution indexed by $\vartheta \in \Theta$. The standard definition of
parametric identification centers around the injectivity of the parametrization
({\em e.g.}, Definition~11.2.2 of \cite{casella2002statistical} and Definition~5.2 of \cite{lehmann2006theory}).

\begin{definition}[Parametric identification]\label{D:parametric-identification}
  The parameter $\vartheta$ of statistical model $\Lambda = \{P_\vartheta, \vartheta \in \Theta\}$
  is said to be identifiable if the function $\vartheta \rightarrow P_\vartheta$ is injective.
\end{definition}
For parametric statistical models, Definition~\ref{D:parametric-identification}
is equivalent to Definition~\ref{D:identifiability-function} with appropriately
chosen statistical universe $\Sall$, observation mapping $\lambda$, and estimand
mapping $\theta$.
\begin{theorem}
  For a parameter set $\Theta$, define the statistical universe to be
  $\Sall = \{(P_{\vartheta}, \vartheta), \vartheta \in \Theta\}$. Let the
  inference and estimand mappings $\lambda, \theta \in G(\Sall)$ be
  $\lambda(S) = P_\vartheta$ and
  $\theta(S) = \vartheta$, respectively. In this setting,
  Definition~\ref{D:identifiability-function} is equivalent to
  Definition~\ref{D:parametric-identification}.
\end{theorem}
\begin{proof}
By construction, the induced binary relation $R_{\theta,\lambda}$ is functional and left-total; therefore, $R_{\theta,\lambda}$ is a function mapping $\vartheta$ to $P_\vartheta$. The conclusion follows from Proposition~\ref{prop:injective}
\end{proof}
One of the classic textbook examples is the identification of the parameters in a linear regression. 

\begin{example}[Linear regression]
   \label{example:linear-regression} 

  Consider a $p$-dimensional random vector $X \sim P(X)$ for 
  some distribution $P_X$ such that $E[X^tX]$ has rank 
  $r < p$, where $E$ denotes the expectation with respect to the law of $X$. Let 
  $P(Y \mid X; \beta,\sigma^2) 
  = \mathcal{N}(X^t\beta, \sigma^2)$, where $\mathcal{N}(\mu,\sigma^2)$ is the
  normal distribution with mean $\mu$ and variance $\sigma^2$, and let
  $P_{\beta,\sigma^2}(X,Y) = P(Y \mid X; \beta, \sigma^2) P(X)$. \\
  \emph{Question:} Are the regression parameters $\beta$ and $\sigma^2$
  identifiable?\\
  \emph{Our framework:} We can establish the identifiability of the parameter $\vartheta = (\beta, \sigma^2)$ from the joint
  distribution $P_{\vartheta}(X,Y)$ by letting
  $\Sall = \{(P_{\vartheta}, \vartheta), \vartheta \in \Theta\}$,
  where $\Theta = \mathbb{R} \times \mathbb{R}^{+}$,
  $\lambda(S) = P_{\vartheta}$, and $\theta(S) = \vartheta$. 
\end{example}

Even in the simple parametric setting, the benefits of the added flexibility of our general formulation become apparent when we ask more subtle questions about identifiability. For instance, using the set up of Example~\ref{example:linear-regression}, suppose we are only  interested in identifying $\beta = \phi(\vartheta)$, rather than the pair $\vartheta = (\beta, \sigma^2)$. The standard Definition~\ref{D:parametric-identification} does not apply here, since $\beta \rightarrow P_\vartheta$ is not a function; that is, each value of $\beta$ is associated with an infinite number of distributions $P_\vartheta$, with different values of the parameter $\sigma^2$. Indeed, the key limitation with existing definitions is that they focuse on the injectivity of a function. In contrast, our framework studies the injectivity of binary relations, which need not be functional. Therefore, Definition~\ref{D:identifiability-function} is directly applicable to studying the identifiability of $\beta$ by replacing $\theta(S) =(\beta,\sigma^2)$ by $\theta(S) = \beta$; generally, it allows us to consider any parameter of the model or combinations of parameters without having to introduce new definitions. 
\begin{example}[Mixtures]
  \label{example:mixtures}
  Let $Y_1 \sim \mathcal{N}(\mu_1, 1)$, let $Y_2 \sim \mathcal{N}(\mu_2, 1)$,
  and $B \sim \text{Bernoulli}(\pi)$.\\
  \emph{Question:} Which of the distributional parameters of $Y = B Y_1 + (1-B) Y_2$ are
  identifiable?\\
  \emph{Our framework:} For $\vartheta = (\mu_1, \mu_2, \pi)$, let $P_\vartheta$ be
  the normal distribution with mean $\pi\mu_1+(1-\pi)\mu_2$ and variance
  $\pi^2 + (1-\pi)^2$. Let $\Sall = \{(P_\vartheta, \vartheta), \vartheta \in \Theta\}$ where
  $\Theta = \mathbb{R}\times \mathbb{R} \times [0,1]$. The observation mapping is
  defined as $\lambda(S) = P_\vartheta$ and the estimand mappings are
  $\theta_{\pi}(S) = \pi$, $\theta_{\mu_1}(S) = \mu_1$,
  $\theta_{\mu_2}(S) = \mu_2$.%
  The question of the identifiability of the parameters can be resolved by studying
  the injectivity of the binary relations $R_{\theta_{\mu_1},\lambda}$,
  $R_{\theta_{\mu_2},\lambda}$, and $R_{\theta_{\pi},\lambda}$ as in
  Definition~\ref{D:identifiability-function}.
\end{example}
Clearly, the three induced binary relation $R_{\theta_{\mu_1},\lambda}$, $R_{\theta_{\mu_2},\lambda}$, and $R_{\theta_{\pi},\lambda}$ are not functional, making Definition~\ref{D:parametric-identification} nonapplicable. Traditionally, authors have tackled this problem by proposing a separate definition for identifying a function of $\vartheta$ (\emph{e.g.}, ~\cite{paulino1994identifiability}[Definition~2.4]); \cite{basu2006identifiability} refers to this as {\em partial identifiability}. Our definition of identification for functions of the parameter agrees with both \cite{paulino1994identifiability} and \cite{basu2006identifiability}, with the added benefit of working directly for both parameters and functions of parameters without requiring additional formulation. 

\subsection{Nonparametric models}
\label{section:nonparams-models}
 Many authors have recognized the limitations of traditional definitions for parametric identifiability (Definition~\ref{D:parametric-identification}) when working with nonparametric models, and have proposed specialiazed frameworks \citep{hurwicz1950generalization,   matzkin2007nonparametric, matzkin2013nonparametric, pearl2009causal}. Consider, for instance, the framework described by \cite{matzkin2007nonparametric} to define identifiability. Let $\Sall$ be the set of all functions and distributions that satisfy the restriction imposed by some model $\mathcal{M}$,  and assume that any $S \in \Sall$ defines a distribution of the observable variables $P(.;S)$. Similar to our general set up, \cite{matzkin2007nonparametric}[Section~3.1] considers a function $\theta : \Sall \to \Theta$ which defines a feature of $S$ we would like to learn about. \cite{matzkin2007nonparametric} then proposes the following definition.

\begin{definition}\label{D:identification-nonparametric-matzkin}
    For $\vartheta_0 \in \theta(S)$, let 
    \[
        \Gamma(\vartheta_0,\Sall) = \{ P(.;S) | S \in \Sall \text{ and } \theta(S) = \vartheta_0)\},
    \]
    be the set of all probability distributions that satisfy the constraints
    of  model $\mathcal{M}$, and are consistent with $\vartheta_0$ and
    $\mathcal{S}$. Then $\vartheta_1 \in \Theta$ is identifiable if for any
    $\vartheta_0 \in \Theta$ such that $\vartheta_0 \ne \vartheta_1$
    \[
        \Gamma(\vartheta_0,\Sall) \cap \Gamma(\vartheta_1,\Sall) = \emptyset
    \]
\end{definition}
This definition of nonparametric identifiability can be obtained as a special
case of our general definition:
\begin{theorem} \label{T:equivalance-nonpar-our}
  Let $\Sall$ be the set of all functions and distributions that satisfy the restriction imposed by some model $\mathcal{M}$. Define $\lambda(S) = P(.;S)$, then Definition \ref{D:identifiability-function} is equivalent to Definition \ref{D:identification-nonparametric-matzkin}.
\end{theorem}

\begin{proof}
  In our notation, we can write
  Definition~\ref{D:identification-nonparametric-matzkin} as:
  \begin{equation}\label{eq:def-rewritten}
    \forall \vartheta_0, \vartheta_1 \in \Theta, \quad
     \vartheta_0 \neq \vartheta_1 \quad \Rightarrow \quad \Gamma(\vartheta_0, \Sall) \cap \Gamma(\vartheta_1, \Sall) = \emptyset,
   \end{equation}
   which is equivalent to
   \begin{flalign*}
     \eqref{eq:def-rewritten}
     &\qquad \iff \qquad
     \neg \bigg(\Gamma(\vartheta_0, \Sall) \cap \Gamma(\vartheta_1, \Sall) = \emptyset\bigg)
     \quad \Rightarrow \quad \neg \bigg( \vartheta_0 \neq \vartheta_1\bigg)\\
     &\qquad \iff \qquad
     \Gamma(\vartheta_0, \Sall) \cap \Gamma(\vartheta_1, \Sall) \neq \emptyset
     \quad \Rightarrow \quad \vartheta_0 = \vartheta_1\\
     &\qquad \iff \qquad
     \exists \ell \in \Lambda: \,\, \vartheta_0 R_{\theta,\lambda} \ell \text{ and } \vartheta_1 R_{\theta,\lambda} \ell
     \quad \Rightarrow \quad \vartheta_0 = \vartheta_1     
   \end{flalign*}
   which is the definition of injectivity
   (see Definition~\ref{D:binary-relation-properties}). The conclusion follows
   from Proposition~\ref{prop:injective}.
\end{proof}
  
Theorem~\ref{T:equivalance-nonpar-our} shows that Matzkin's nonparametric identification definition is a special case of our more general framework, with a specific choice of statistical universe $\Sall$ and observation mapping $\lambda$. We now provide three examples that cannot be addressed with Definition~\ref{D:identification-nonparametric-matzkin} and require the additional flexibity afforded by Definition~\ref{D:identifiability-function}.

\begin{example}[Fixed margins problem]\label{example:fixed-margins}
  Consider two distributions $P_X(X)$ and $P_Y(Y)$, and denote by $P_{XY}(X,Y)$ their
  joint distribution. The fixed margin problem \citep{frechet1951tableaux} asks what 
  information the marginal distributions $P_X$ and $P_Y$ contain about the joint
  distribution $P_{XY}$.\\
  \emph{Question:} Is $P_{XY}$ identifiable from $P_X$ and $P_Y$?\\
  \emph{Our framework:} Let $\Sall$ be a family of joint distributions for $X$ and $Y$.
  Let $\lambda(S) = (P_X, P_Y)$ and let $\theta(S) = P_{XY}$. The question of the
  identifiability of $P_{XY}$ from $P_X$ and $P_Y$ can be answered by studying the
  injectivity of the induced mapping $R_{\theta, \lambda}$ as in
  Definition~\ref{D:identifiability-function} (see Section~\ref{section:copulas}
  for a detailed treatment).
\end{example}
In the first example, Definition~\ref{D:identification-nonparametric-matzkin}
falls short by not allowing observation mappings of the form
$\lambda(S) = (P_X, P_Y)$ -- a problem transparently addressed by our
definition. The following example describes another setting in which the
same issue arises. 

\begin{example}[Missing data]\label{ex:missing-data}
  If $Y$ is a random variable representing a response of interest, let $Z$ be a
  missing data indicator that is equal to $1$ if the response $Y$ is observed, and $0$
  otherwise. The observed outcome of interest is then drawn from $P(Y \mid Z=1)$.\\
  \emph{Question:} Is the distribution of the missing outcomes $P(Y \mid Z=0)$
  identifiable from that of the observed outcomes $P(Y \mid Z=1)$?\\
  \emph{Our framework:} Let $\Sall$ be a family of joint distributions for $Z$ and $Y$,
  and define $\lambda(S) = (P(Y \mid Z=1), P(Z))$, and $\theta(S) = P(Y \mid Z=0)$. The
  question can be answered by studying the
  injectivity of the induced mapping $R_{\theta, \lambda}$ as in
  Definition~\ref{D:identifiability-function}.
\end{example}
Example~\ref{ex:missing-data} shows that $\theta$ need not be the identity
function: here for instance, we are interested in the conditional distribution
$\theta(S) = P(Y \mid Z=0)$. In fact, $\theta(S)$ does not even need to be a
distribution: in the following example, it is a conditional expectation. 
\begin{example}[Ecological regression]\label{example:ecological-regression}
  Ecological inference is concerned with extracting individual-level information
  from aggregate data \citep{king2013solution}. An instance of the ecological inference
  problem is the ecological regression problem \citep{cross2002regression} which can be
  summarized as follows: suppose we know the distributions $P(Y \mid X)$ and
  $P(Z \mid X)$. What information does this give us about the expectation
  $E[Y \mid X, Z]$?\\
  \emph{Question:} Is $E[Y \mid X, Z]$ identifiable from $P(Y \mid X)$ and
  $P(Z \mid X)$.\\
  \emph{Our framework:} Let $\Sall$ be a family of joint distributions for $Y$, $X$,
  and $Z$. Define $\lambda(S) = (P(Y\mid X), P(Z\mid X))$ and
  $\theta(S) = E[Y \mid X, Z]$. The question can be answered by studying the
  injectivity of the induced mapping $R_{\theta, \lambda}$ as in
  Definition~\ref{D:identifiability-function}.
\end{example}
\subsection{Identification in finite populations}
\label{section:id-finite-pop}

The examples presented so far asked questions about identifiability in the
context of statistical models: the statistical universe, estimand mappings
and observation mappings involved entire distributions (or summaries of
distributions). Implicitly, this corresponds to the ``infinite observations''
perspective of \cite{koopmans1949identification} quoted in introduction. The
missing data problem of Example~6, for instance, asks about the
identifiability of $P(Y \mid Z=0)$ from $P(Y \mid Z=1)$. Consider instead a
finite population of $N$ units and denote by $Y_i$ an outcome of interest for
unit $i = 1, \ldots, N$. Suppose we only observe $Y_i$ if $Z_i = 1$ and
that the outcome is missing when $Z_i=0$; formally, we observe
$\{Y^\ast_i, Z_i\}_{i=1}^N$, where

\begin{equation*}
  Y_i^\ast = \begin{cases}
    Y_i & \text{if } Z_i = 1 \\
    \ast & \text{otherwise}.
    \end{cases}
\end{equation*}
This is the finite population analog to Example~6. A natural identification
question would be: can we identify $\tau = \overline{Y}=\sum_{i=1}^N Y_i$ from
$\{Y_i^\ast, Z_i\}$? Neither Definition~\ref{D:parametric-identification} nor
Definition~\ref{D:identification-nonparametric-matzkin} apply in this setting,
since there are no statistical models (parametric or nonparametric) involved --- 
and yet the identifiability question makes intuitive sense.

\cite{manski2009identification} addresses the problem by introducing a sampling model and applying concepts of identification from nonparametric statistical models. Specifically, let $I \sim \text{Unif}(\{1, \ldots, N\})$ be a random variable selecting a unit uniformly at random in the population. The population average $\tau$ can be rewritten as $E_I[Y_I]$, where $E_I$ is the expectation with respect to the distribution of $Y_I$ induced by $P(I)$. Manski's approach allows us to work with $P_I(Y_I^\ast, Z_I)$ instead of $\{Y_i^\ast, Z_i\}_{i\in I}$, and to rephrase the problem as whether $E_I[Y_I]$ is identifiable from $P_I(Y_I^\ast, Z_I)$ --- the results of  Section~\ref{section:nonparams-models} are now directly applicable. A downside of this approach, however, is that it changes the objective somewhat. Indeed, while $\tau$ and $E_I[Y_I]$ refer to the same quantity, $P_I(Y_I^\ast, Z_I)$ contains strictly less information than $\{Y_I^\ast, Z_I\}$, making it impossible to formulate some seemingly simple identification questions; such as, whether $Y_1$ can be identified from $\{Y_i^{\ast}, Z_i\}$.

By contrast, our general definition accomodates this setting
by simply specifying appropriate $\Sall$, $\theta$, and $\lambda$. Let $\Sall_Y$ be
a set of possible outcome vectors, $\Sall_Z$ be a set of possible assignment
vectors, and define $\Sall = \Sall_Y \times \Sall_Z$. An element $S \in \Sall$
is then a pair $(\{Y_i\}_{i=1}^N, \{Z_i\}_{i=1}^N)$. Finally, define the
observation mapping $\lambda(S) = \{Y_i^\ast, Z_i\}_{i=1}^N$ and
the estimand mapping as, for instance,  $\theta(S) = \overline{Y}$ or
$\theta(S) = Y_1$. The following example illustrates our framework in a slightly
more involved finite-population setting.
\begin{example}[Population identification of causal effects]
\label{example:id-causal-effects} With $N$ units, let  
each unit be assigned to one of two treatment interventions, $Z_i =1$ for
treatment and $Z_i=0$ for control. Under the stable unit treatment value
assumption \citep{rubin1980randomization} each unit $i$ has two potential
outcomes $Y_i(1)$ and $Y_i(0)$, corresponding to the outcome of unit $i$ under
treatment and control, respectively. For each unit $i$, the observed outcome is 
$Y^{\ast}_i = Y_i(Z) = Y_i(1) Z_i + Y_i(0) (1-Z_i)$.
Let $Y(1) = \{Y_1(1), \ldots, Y_N(1)\}$ and $Y(0) = \{Y_1(0), \ldots, Y_N(0)\}$ be
the vectors of potential outcomes and $Y = (Y(1), Y(0))$.\\
\emph{Question:} Is $\tau(Y) = \overline{Y(1)} - \overline{Y(0)}$ identifiable from the
observed data $(Y^\ast, Z)$.\\
\emph{Our framework:} Let $\Sall_Y = \mathbb{R}^N \times \mathbb{R}^N$ be the set of all
possible values for $Y$, $\Sall_Z = \{0,1\}^N$ and $\Sall = \Sall_Y \times \Sall_Z$.
Take $\theta(S) = \tau(Y)$ and $\lambda(S) = (Y^\ast, Z)$ as the estimand and observation mapping, respectively. 
The question is then answerable by studying the injectivity of the induced binary relation $R_{\theta, \lambda}$ as in Definition~\ref{D:identifiability-function}.
\end{example}

\section{Identification analysis}\label{section:id-analysis}

So far, we have shown how a variety of identification questions can be
formulated in our framework, but we have said nothing about how they can
answered. Identification analysis is a three-steps process for answering
such questions. The first step is to establish whether $\theta$ is identifiable or not (Section~\ref{section:ia-identification}). For not idenfiable $\theta$, the second step is to determine its identification region (Section~\ref{section:ia-region}). The third step is to incorporate different assumptions and assessing their impact on the structure of the identification region (Section~\ref{section:ia-assumptions}).

\subsection{Determining if $\theta$ is identifiable}
\label{section:ia-identification}
The most direct---but usually challenging---approach to determine if $\theta$ is $R_{\theta, \lambda}$-identifiable is to use Definition~\ref{D:identifiability-function}. A simpler alternative is to instead show that $\theta(S)$ is a function of $\lambda(S)$ for all $S\in \Sall$; ensuring that each $\lambda(S)$ is in relation with a single $\theta(S)$.

\begin{proposition}\label{proposition:id-function}
	If there exists a funtion $f: \Lambda \rightarrow \Theta$ 
    such that $\theta(S) = f(\lambda(S))$ for all $S \in \Sall$, then 
    $\theta$ is $R_{\theta, \lambda}$-identifiable.
\end{proposition}
\begin{proof}
  Fix $\ell_0 \in \Lambda$, let $\vartheta_0 = f(\ell_0)$ and consider $\Sall_0 =\{S \in \Sall: \lambda(S) = \ell_0\}$. For any $S \in \Sall_0$ we have that $\theta(S) = f(\lambda(S)) = f(\ell_0) = \vartheta_0$; therefore, $\theta$ is identifiable at $\ell_0$. Since this holds for any $\ell_0 \in \Lambda$, $\theta$ is $R_{\theta, \lambda}$-identifiable.
\end{proof}

To illustrate the concepts in this section, we draw on an extended treatment of Example~\ref{ex:missing-data}, which discusses identification in missing data. It is, for instance, easy to show that the marginal probability $\theta_1(S) = P(Z = 1)$ is identifiable by noticing that it can be written as a function of $\lambda(S) = (P(Y | Z=1), P(Z=1))$ and applying Proposition~\ref{proposition:id-function}. 

In many applications Proposition~\ref{proposition:id-function} is still difficult to apply directly either because $\theta(S)$ is a complicated function of $\lambda(S)$, or because it is not even a function of $\lambda(S)$. Either way, it is often better to first break up the estimand mapping $\theta(S)$ into simpler pieces, establish identifiability for each of them, and then leverage the fact that a function of identifiable quantities is itself identifiable.
\begin{proposition}
\label{proposition:id-operations}
Let $\theta, \theta_1, \theta_2 \in G(\Sall)$, and $f$ a 
function such that:
\begin{equation*}
  \forall S \in \Sall, \qquad 
  \theta(S) = f(\theta_1(S), \theta_2(S)).
\end{equation*}
If $\theta_1$ is $R_{\theta_1,\lambda}$-identifiable and $\theta_2$ is
$R_{\theta_2, \lambda}$-identifiable, then $\theta$ is
$R_{\theta, \lambda}$-identifiable. This trivially generalizes to
$\theta_1,\dots , \theta_T \in G(\Sall)$.
\end{proposition}
\begin{proof}
  Fix $\ell_0 \in \Lambda$ and let $\Sall_0 = \{S \in \Sall: \lambda(S) = \ell_0\}$.
  Since $\theta_1$ is $R_{\theta_1,\lambda}$-identifiable and $\theta_2$ is
  $R_{\theta_2,\lambda}$-identifiable, then by Definition~\ref{D:identifiability-function},
  there exists $\vartheta_1 \in \Img(\theta_1)$ and $\vartheta_2 \in \Img(\theta_2)$ such
  that:
  \begin{equation*}
    \forall S \in \Sall_0, \quad \theta_1(S) = \vartheta_1 \,\,
    \text{ and } \,\, \theta_2(S) = \vartheta_2
  \end{equation*}
  and so:
  \begin{equation*}
    \forall S \in \Sall_0, \quad \theta(S) = f(\theta_1(S), \theta_2(S))
    = f(\vartheta_1, \vartheta_2) \equiv \vartheta_0 \in \Img(\theta).
  \end{equation*}
\end{proof}
In our missing data example, the quantity of interest is the average response
$\theta(S) = E[Y]$. Applying the strategy described above, we can write
it as a function of simpler quantities:
\begin{equation*}
  \underbrace{E[Y]}_{\theta(S)} =
  \underbrace{E[Y \mid Z=1] P(Z=1)}_{\theta_a(S)}
  +
  \underbrace{E[Y \mid Z=0] P(Z=0)}_{\theta_b(S)}.
\end{equation*}
Starting with the first term $\theta_a$,
\begin{equation*}
  \theta_a(S) =
  \underbrace{E[Y \mid Z=1]}_{\theta_2(S)}
  \underbrace{P(Z=1)}_{\theta_1(S)},
\end{equation*}
we have already shown that $\theta_1$ is identifiable. Another application
of Proposition~\ref{proposition:id-function} establishes the identifiability
of $\theta_2(S) = E[Y \mid Z=1]$ and
Proposition~\ref{proposition:id-operations} stitches these results together
to establish the identifiability of $\theta_a$. 
The second term $\theta_b$, which also decomposes into two parts
\begin{equation*}
  \theta_b(S) =
  \underbrace{E[Y \mid Z=0]}_{\theta_3(S)}
  \underbrace{P(Z=0)}_{1-\theta_1(S)},
\end{equation*}
is not identifiable because: although $1-\theta_1(S)$ is identifiable (by Proposition~\ref{proposition:id-function}), $\theta_3$ generally is not. To see this, consider the following simple counter-example. Suppose $Y$ is binary, the elements of $\Sall$ are then of the form $S = \{P(Y=1 \mid Z=1)=\alpha, P(Y=1 \mid Z=0)=\beta, P(Z=1)=\gamma\}$. Recall the observation mapping $\lambda(S) = (P(Y=1 \mid Z=1)=\alpha, P(Z=1)=\gamma)$. Fix $\ell_0 = (\alpha_0, \gamma_0) \in \Lambda$ and define $S_1 = (\alpha_0, \beta_1=0, \gamma_0)$ and $S_2 = (\alpha_0, \beta_2 = 1, \gamma_0)$. By construction, $ \theta_3(S_1)=0\ne1= \theta_3(S_2)$ while $\lambda(S_1) = \lambda(S_2)$; applying Definition~\ref{D:identifiability-function} shows that $\theta_3$ is not identifiable. The counter-example illustrates that $\theta_b$ is generally not identifiable which implies that $\theta$ is not identifiable (except when there is no missing data $\theta_1(S) = 1$).

\subsection{Finding $\theta$'s identification region}
\label{section:ia-region}
By Definition \ref{D:identifiability-function}, we know that if an estimand mapping $\theta$ is not identifiable at $\ell_0$, then there exists at least two values $\vartheta_1, \vartheta_2 \in \Theta$ such that $\vartheta_1 R_{\theta,\lambda} \ell_0$ and $\vartheta_2 R_{\theta,\lambda} \ell_0$. The second step in identification analysis is to determine the set of all $\vartheta \in \Theta$ such that $\vartheta R_{\theta,\lambda} \ell_0$.
This set is generally called the identification region of $\theta$ at $\ell_0$
\citep{manski1990nonparametric,imbens2004confidence,romano2008inference}.
\begin{definition}\label{D:identification-region}
Consider $\Sall$ and $\theta, \lambda \in G(\Sall)$. We define the 
identification region of $\theta$ at $\ell_0 \in \Lambda$ as:
\begin{equation*}
  H\{\theta; \ell_0\} \equiv R_{\theta,\lambda}^{-1}(\ell_0) \subseteq \Theta
\end{equation*}
where:
\begin{equation*}
  R_{\theta,\lambda}^{-1}(\ell_0) \equiv \{\vartheta \in \Theta:
  \vartheta R_{\theta,\lambda} \ell_0\}
  = \{\theta(S): S \in \Sall, \lambda(S) = \ell_0\}
\end{equation*}
is the pre-image of $\ell_0$ in $\Theta$.
\end{definition}
Informally, the identification region is the set of all values
of the estimand that are equally compatible with the observation $\ell_0$. If
an estimand mapping $\theta$ is identifiable at $\ell_0$, then a single estimand
is compatible with $\ell_0$ and the identification region reduces to a
singleton.
\begin{proposition}
  For $\theta, \lambda \in G(\Sall)$, we have that
  \begin{itemize}
  \item $\theta$ is $R_{\theta,\lambda}$-identifiable at $\ell_0 \in \Lambda$
    if and only if $H\{\theta; \ell_0\} = \{\vartheta_0\}$ for some
    $\vartheta_0 \in \Theta$,
  \item $\theta$ is $R_{\theta,\lambda}$-identifiable everywhere if and only if
    $H\{\theta; \ell_0\}$ is a singleton for every $\ell_0 \in \Lambda$.
    \end{itemize}    
\end{proposition}
\begin{proof}
  We only need to prove the first element, since the second follows by
  definition. Fix $\ell_0 \in \Lambda$ and define
  $\Sall_0 = \{S \in \Sall: \lambda(S) = \ell_0\}$. Suppose that $\theta$ is
  $R_{\theta,\lambda}$-identifiable at $\ell_0$. By
  Definition~\ref{D:identifiability-function}, there exits $\vartheta_0$ such
  that $\theta(S) = \vartheta_0$ for all $S \in \Sall_0$. The identification region is then
  \begin{flalign*}
    H\{\theta; \ell_0\}
    &= \{\theta(S): S\in \Sall, \lambda(S) = \ell_0\}\\
    &= \{\theta(S): S \in \Sall\} \\
    &= \{\vartheta_0\}.
  \end{flalign*}
  This completes the first part of the proof. Now suppose that $\theta$ is such
  that $H\{\theta; \ell_0\} = \{\vartheta_0\}$ for some $\vartheta_0 \in \Theta$.
  Then, by Definition~\ref{D:identification-region}, $\theta(S) = \vartheta$ for
  all $S \in \Sall$ such that $\lambda(S) = \ell_0$ which is the definition of
  identifiability (Definition~\ref{D:identifiability-function}). 
\end{proof}
To illustrate the concept of an identification region, consider Example~\ref{ex:missing-data}, with the additional information that $Y \in \mathbb{R}$ (making
$\Theta = \mathbb{R}$), and $P(Z = 1) \ne 1$. Fixing
$\ell_0 = (P^{(0)}(Y \mid Z=1), P^{(0)}(Z=1))\in \Lambda$ we need to check for
every $\vartheta \in \Theta$ whether it belongs to the identification region of
$\ell_0$. The previous section showed that both $\theta_1 = E[Y|Z=1]$ and $\theta_2=P(Z=1)$
are identifiable so let $\vartheta_1$ and $\vartheta_2$ be the unique elements
of $\Theta$ such that $\vartheta_1 R_{\theta_1,\lambda} \ell_0$ and
$\vartheta_2 R_{\theta_2,\lambda} \ell_0$. Now take $P(Y \mid Z = 0)$ to be
any distribution with expectation
$E[Y \mid Z=0] = (\vartheta - \vartheta_2 \vartheta_1) / (1-\vartheta_1)$. By
construction, $S = (P^{(0)}(Y \mid Z=1), P(Y \mid Z=0), P^{(0)}(Z=1))$ is
such that $\lambda(S) = \ell_0$ and $\theta(S) = \vartheta$, therefore 
$\vartheta \in H\{\theta;\ell_0\}$. Since this is true for all
$\vartheta \in \mathbb{R}$, we conclude that $H\{\theta; \ell_0\} = \mathbb{R}$.

Generally, we interpret $\Theta$ as the values of the estimand \emph{a priori} possible, and $H\{\theta; \ell_0\}$ as the values of the estimands
compatible with the observation $\ell_0$. If $H\{\theta; \ell_0\} = \Theta$ then, not only is $\theta$ not identifiable, but $\ell_0$ carries
no information about the range of plausible values for $\theta$. We call such a quantity {\em strongly non-identifiable}.

\begin{definition}[Strong non-identifiability]
  A mapping $\theta \in G(\Sall)$ is said to be strongly non-identifiable if 
  \begin{equation*}
    \forall \ell \in \Lambda, \quad 
    H\{\theta; \ell\} = \Theta
  \end{equation*}
  We will usually write $H\{\theta; \ell\} = H\{\theta\}$. 
\end{definition}
When discussing how to establish identifiability, we argued that it was often helpful to break the quantity of interest into simpler pieces and work separately on each piece before re-combining them --- the same strategy applies to the identification region. The following proposition gives simple rules for combining identification regions.
\begin{proposition}\label{proposition:id-region-operations}
  Let $\lambda, \theta_1, \theta_2 \in G(\Sall)$. Let $\theta \in G(\Sall)$, and consider 
  a function $f$ such that:
    \begin{equation*}
    	\forall S \in \Sall, \theta(S) = f(\theta_1(S), \theta_2(S))
    \end{equation*}
    Then if $\theta_1$ is $R_{\theta_1,\lambda}$-identifiable, the 
    identification region of $\theta$ at $\ell_0 \in \Lambda$ 
    is
    \begin{equation*}
      H\{\theta; \ell_0\} = \{f(\vartheta_0, \vartheta): 
      \, \vartheta \in H\{\theta_2; \ell_0\}\},
    \end{equation*}
    where $H\{\theta_2; \ell_0\}$ is the identification region of 
    $\theta_2$ with respect to $R_{\theta_2,\lambda}$, at $\ell_0$,
    and $\{\vartheta_0\} = R^{-1}_{\theta_1,\lambda}(\ell_0)$.
\end{proposition}
\begin{proof}
  By Definition~\ref{D:identification-region},
  \begin{equation*}
    H\{\theta; \ell_0\}
    = \{f(\theta_1(S),\theta_2(S)): S\in\Sall,\lambda(S) = \ell_0\} 
  \end{equation*}
  Since $\theta_1$ is $R_{\theta_1,\lambda}$-identifiable, by
  Definition~\ref{D:identifiability-function}, there exists
  $\vartheta_0 \in \Theta$ such that $\theta_1(S) = \vartheta_0$
  for all $S \in \Sall$ such that $\lambda(S) = \ell_0$. The identification region is then,
  \begin{equation*}
    H\{\theta; \ell_0\}
    = \{f(\vartheta_0,\theta_2(S)): S\in\Sall,\lambda(S) = \ell_0\}.
  \end{equation*}
  Now notice that
  $H\{\theta_2(S), \ell_0\} = \{\theta_2(S): S\in \Sall, \lambda(S) = \ell_0\}$
  by definition, and so:
  \begin{flalign*}
    H\{\theta; \ell_0\}
    &= \{f(\vartheta_0, \vartheta): \vartheta \in \{\theta_2(S):
    S \in \Sall, \lambda(S) = \ell_0\}\}\\
    &= \{f(\vartheta_0, \vartheta): \vartheta \in H\{\theta_2; \ell_0\}\}
  \end{flalign*}
\end{proof}
When we proved that $H\{\theta; \ell_0\} = \mathbb{R}$ earlier in this section,
we relied implicitly on this result. In fact, the decomposition
\begin{equation}
  \underbrace{E[Y]}_{\theta(S)} =
  \underbrace{E[Y \mid Z=1]}_{\theta_2(S)} \underbrace{P(Z=1)}_{\theta_1(S)}
  -
  \underbrace{E[Y \mid Z=0]}_{\theta_3(S)} \underbrace{P(Z=0)}_{\theta_4(S)}
\end{equation}
has an additional property of interest: the terms $\theta_1, \theta_2$
and $\theta_4$ are identifiable while $\theta_3$ is strongly non-identifiable.
Intuitively, this factorization isolates the non-identifiable part of $\theta$
in a single term. We call this factorization a {\em reduced form}.
\begin{proposition}\label{def:reduced-form} Fix $\mathcal{S}$, $\lambda$, and $\theta$ such that we can factorize as $\theta(S) = f(\{\theta_k(S)\}_{k=1}^K, \theta_\ast(S))$ for some $f$ and $\{\theta_k\}_k^K\in G(\mathcal{S})$. If
  \begin{enumerate}
  \item $\theta_k$ is $R_{\theta_k, \lambda}$-identifiable 
    at $\ell_0$ for $k=1,\dots,K$, and
  \item $\theta_\ast$ is strongly
    non-$R_{\theta_\ast,\lambda}$-identifiable at $\ell_0$.
  \end{enumerate}
  then the \emph{reduced form} of $H\{\theta; \ell_0\}$ is
  \begin{equation*}
    H\{\theta; \ell_0\} 
    = \bigg\{f(\{\vartheta_k\}_{k=1}^K, \vartheta), \,\,
    \vartheta \in H\{\theta_\ast\}\bigg\},
  \end{equation*}

  where $\vartheta_k$ is the unique element of $\Theta$ such that $\vartheta_k R_{\theta_k,\lambda} \ell_0$, for $k=1,\dots, K$. 
\end{proposition}
\begin{proof}
  The proof follows the same lines as that of
  Proposition~\ref{proposition:id-region-operations}, the difference being that since $\theta_\ast$ is strongly-non identifiable,
  we have $H\{\theta_\ast, \ell_0\} = H\{\theta_\ast\}$.
\end{proof}

In our running example, both $\theta_1(S) = P(Z=1)$ and $\theta_2(S) = E[Y|Z=1]$ are identifiable, whereas $\theta_3(S) = E[Y|Z=0]$ is strongly non-identifiable. The reduced form of the identification region for $\theta = E[Y]$ at $\ell_0 \in \Lambda$ is then
\begin{equation*}
  H\{\theta;\ell_0\} = \{(\vartheta_2 \vartheta_1 + \vartheta (1 - \vartheta_1)), \vartheta \in H\{\theta_3\}\},
\end{equation*}
where  $\vartheta_1$ and $\vartheta_2$ are the unique elements of $\Theta$ such that $\vartheta_1R_{\theta_1,\lambda} \ell_0$ and $\vartheta_2R_{\theta_2,\lambda} \ell_0$, respectively.
Since $\theta_3$ is strongly non-identifiable $H\{\theta_3;\ell_0\} = \mathbb{R}$ and,
therefore, $H\{\theta; \ell_0\} = \mathbb{R}$, when $P(Z=1) \neq 1$.

\subsection{Incorporating assumptions}
\label{section:ia-assumptions}

In the derivation of the identification region of $H\{\theta; \ell_0\} = \mathbb{R}$, we made no assumption about the outcomes $Y$ (only requiring them to be real numbers). Suppose that we assume $Y \in [0,1]$, how does this
affect the identification region of $\theta$ at $\ell_0$? That is the type of
question that the third step of identification analysis seeks to answer. In our
framework, we formalize assumptions as functions inducing restrictions on the
statsistical universe $\Sall$.
\begin{definition}[Assumption]
	An assumption is a function $A: \Sall \rightarrow \mathbb{R}$. The set 
    $\A = \{S\in \Sall: A(S) = 0\}$ is called the subset of $\Sall$ 
    satisfying assumption $A$.
\end{definition}
To incorporate assumptions in our framework, we augment our notation with a
superscript $A$; defining, for instance, $R^A_{\theta,\lambda}$ to be the
restriction of $R_{\theta,\lambda}$ to the set $\A$. In general, the purpose
of an assumption is to make $R^A_{\theta,\lambda}$ ``closer'' to injective
(Figure \ref{fig:assumptions}, provides an intuitive visualization). The restricted identification region is
then a subset of the full identification region,
\begin{equation*}
  H^A\{\theta; \ell_0\} = \{\theta(S): S\in \A, \lambda(S) = \ell_0\}
  \subseteq H\{\theta; \ell_0\}.
\end{equation*}
In particular, an assumption makes $\theta$ identifiable when
$H^A\{\theta; \ell_0\}$ is a singleton. For instance, continuing with our
missing data example, let $A(S) = \delta(P(Y,Z), P(Y)P(Z))$ where $\delta$
is the total variation distance. This specification is equivalent to
assuming that $Y$ and $Z$ are independent\footnote{\cite{bimetrkia2015Mealli} call this assumption
missing always completely at random.}. Under
this assumption, it is easy to verify that $\theta(S) = E[Y]$ is identifiable. 
\begin{figure}\label{fig:assumptions}
  \centering
  \includegraphics[scale=0.5]{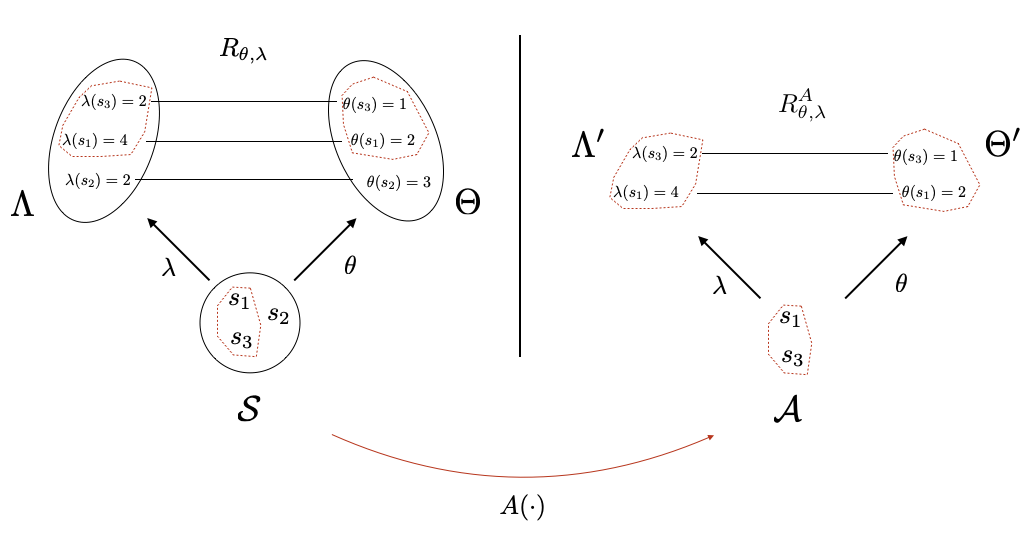}
    \caption{Illustration of an assumption. On the left, 
    $R_{\theta, \lambda}$ is not injective. On the right, 
    $R^{A}_{\theta, \lambda}$ is injective making $\theta$ identifiable.}
\end{figure}
Following \cite{manski2009identification} we distinguish two types of assumptions.
\begin{definition}
	Define $\text{Img}_A(\lambda) = \{\lambda(S): S\in \A\}$. If 
    $\text{Img}_A(\lambda) = \text{Img}(\lambda)$, then assumption 
    A is said to be a priori irrefutable. Otherwise, it is called 
    a priori refutable.
\end{definition}
We use the term \emph{a priori} to stress that we can determine if an assumption
is refutable before observing any data. Once we observe $\lambda(S) = \ell_0$
an \emph{a priori} refutable assumption is either refuted, or it is not, depending on whether
$\ell_0 \in \text{Img}_A(\lambda)$; in other words, an assumption may be
\emph{a priori} refutable and yet not be refuted by a given $\ell_0$. 
In the context of our running example, the assumption that $Y \in [0,1]$ is
\emph{a priori} refutable --- observing data points outside of this interval would
refute the assumption.
Under this assumption, it is easy to verify that
$H^{A}\{E[Y|Z=0]; \ell_0\} = [0,1]$ and, therefore, the reduced form of the
identification region of $\theta(S) = E[Y]$ at $\ell_0$ is:
%
\begin{equation*}
  H^{A}\{\theta;\ell_0\} = \{(\vartheta_2 \vartheta_1 + \vartheta (1 - \vartheta_1)), \vartheta \in [0,1]\},
\end{equation*}
which provides a natural bound for $E[Y]$. 
\section{Case studies}
\label{section:case-studies}


\subsection{Fixed margin problem}
\label{section:copulas}

Continuing with the setup of Example~\ref{example:fixed-margins} with the added simplifying assumption that all the probabilities are continuous, let $F_X$, $F_Y$ and $F_{XY}$ be the cumulative density functions (CDFs) of $P_X$, $P_Y$, and $P_{XY}$ respectively. Recall that $S = P_{XY}$ and $\Sall$ is the set of all continuous joint probability distributions. The observation mapping is $\lambda(S) = (F_X, F_Y)$ and the estimand mappings are 
$\theta_X(S) = F_X$, $\theta_Y(S) = F_Y$ and $\theta(S) = F_{XY}$.

We begin by establishing that $\theta$ is not everywhere $R_{\theta,\lambda}$-identifiable. Let $\Ccal$ be the set of all copulas and 
let $\ell_0 = (F_X^{(0)}, F_Y^{(0)})$. By Sklaar's theorem, for all
$C \in \Ccal$, the function
\begin{equation*}
  F_{XY}^{(0,C)}: (x,y) \mapsto C(F_X^{(0)}(x), F_Y^{(0)}(y))
\end{equation*}
is a valid joint CDF with margins $F_X^{(0)}$ and $F_Y^{(0)}$. In
particular, let $C_1 \neq C_2 \in \Ccal$ and $S_1 = F_{XY}^{(0,C_1)} \in \Sall$, 
$S_2 = F_{XY}^{(0,C_2)} \in \Sall$. Then
$\lambda(S_1) = \lambda(S_2) = (F_X^{(0)}, F_Y^{(0)}) = \ell_0$ by construction,
but $\theta(S_1) \neq \theta(S_2)$ since $C_1 \neq C_2$. That is, $\theta$ is not
$R_{\theta,\lambda}$-identifiable at $\ell_0$. 

Next, we will derive the identification region of $\theta$ at $\ell_0$---the set of all CDFs of continuous joint distributions with margins
$F_X^{(0)}$ and $F_Y^{(0)}$---in its reduced form. By Sklaar's theorem, for any joint CDF $F_{XY}$
with margins $F_X$ and $F_Y$, there exists a unique (since we focus on
continuous probabilities) copula $C \in \Ccal$ such that:
\begin{equation*}
  F_{XY}(x,y) = C(F_X(x), F_Y(y)), \qquad \forall x,y.
\end{equation*}
Denote by $\theta^\ast \in G(\Sall)$ the function that maps each $S = P_{XY}$ to
its unique associated copula, and let $f$ be the function that maps
$(\{F_X, F_Y\}, C)$ to the valid joint CDF $F_{XY}(x,y) = C(F_X(x), F_Y(y))$. In our
notation, we have:
\begin{equation*}
  \theta(S) = f(\{\theta_X(S), \theta_Y(S)\}, \theta^\ast(S)).
\end{equation*}
Since $\theta_X$ and $\theta_Y$ are identifiable, we can write:
\begin{equation}\label{eq:fixed-margin-reduced}
  H\{\theta; \ell_0\} = \bigg\{f(\{\vartheta_X^{(0)}, \vartheta_Y^{(0)}\}, \vartheta),
  \quad
  \vartheta \in H\{\theta^\ast; \ell_0\}\bigg\}
\end{equation}
But it is easy to verify that $H\{\theta^\ast; \ell_0\} = H\{\theta^\ast\} = \Ccal$.
Indeed, reasoning by contradiction, suppose that
$\exists C_0 \in \Ccal \backslash H\{\theta^\ast; \ell_0\}$. Let
$F_{XY}^{(0)} = C_0(F_X^{(0)}, F_Y^{(0)})$. By Sklaar's theorem,
$F_{XY}^{(0)} \in \Sall_0$. But then by definition
$C_0 = \theta^\ast(F_{XY}^{(0)}) \in H\{\theta^\ast; \ell_0\}$ which
is a contradiction. Therefore $\Ccal = H\{\theta^\ast; \ell_0\}$
That is, $\theta^\ast$ is strongly nonidentifiable. So
Equation~\ref{eq:fixed-margin-reduced} is the reduced form representation of
the identification region of $\theta$. 
From Theorem~2.2.3 of
\cite{nelsen1999introduction}, we have that
\begin{equation*}
  	W(x,y) \leq C(x,y) \leq M(x,y) 
\end{equation*}
with $W(x,y) = \max(x+y-1, 0)$ and $M(x,y) = \min(x,y)$. Since $W$ and 
$M$ are both copulas, we have:
\begin{equation*}
  W(F_X(x), F_Y(y)) = \min_{\vartheta \in H\{\theta; \ell_0\}} \vartheta(x,y)
  \quad \leq \quad  \vartheta(x, y)
  \quad \leq \quad  \max_{\vartheta \in H\{\theta; \ell_0\}} \vartheta(x,y) = M(F_X(x), F_Y(y)) 
\end{equation*}
for all $\vartheta \in H\{\theta; \ell_0\}$. This corresponds exactly to the
Hoeffding-Fr\'echet bounds.

We now assume that the joint distribution is a non-degenerate bivariate normal.
Formally, we define the function $A(S) = 0$ iff $S = P_{XY}$ is bivariate normal
and $A(S) = 1$ otherwise. In this case, copulas are of the
form $C_{\rho}(u, v) = \Phi_{\rho}(\Phi^{-1}(u), \Phi^{-1}(v))$ where $\Phi$
is the standard normal $CDF$ and $\Phi_\rho$ is the CDF of a bivariate normal
with mean zero, unit variance, and correlation $\rho$. Now let
$\ell_0 = (F_X^{(0)}, F_Y^{(0)})$ a pair of univariate normal distributions
with parameters $(\mu_X^{(0)}, \sigma_X^{(0)})$ and
$(\mu_Y^{(0)}, \sigma_Y^{(0)})$ respectively, and define
$\tau = (\mu_X^{(0)}, \sigma_X^{(0)}, \mu_Y^{(0)}, \sigma_Y^{(0)})$.
In this setting, we can show by the same reasoning as above that the parameter
$\tau$ is identifiable while the parameter $\rho$ is strongly non-identifiable.
The reduced form of the identification region of $\theta$ is
\begin{equation*}
  H^A\{\theta; \ell_0\} = \{\Phi_{\tau,\rho}, \rho \in [-1,1]\}.
\end{equation*}
where $\Phi_{\tau, \rho}$ is the CDF of the bivariate normal distribution with
parameters $\tau$ and correlation $\rho$. 

If we further assume that the
joint distribution is a non-degenerate bivariate normal with correlation
$\rho_0$ (corresponding to assumption $A_0$) then the identification region
reduces to $H^{A_0}\{\theta; \ell_0\} = \{\Phi_{\tau, \rho_0}\}$; that is,
$\theta$ is identifiable.

\begin{remark}
  The assumption $A$ is an example of a priori refutable assumption. Indeed,
  assumption A implies that the marginal distributions $F_X$ and $F_Y$ are
  themselves normal, so observing non-gaussian marginals would refute the
  assumption.
\end{remark}

\subsection{Causal inference}
For an infinite population of experimental units, we are interested in studying the relative effectiveness of a treatment, denoted by $Z=1$, relative to a control, denoted by $Z=0$, on an outcome of interest. Under the potential outcomes framework, each unit has two outcomes\footnote{We implicitly assume the stable unit treatment value assumption \cite{rubin1980randomization}.} corresponding to what would happen if the unit receives treatment, $Y(1)\in \mathbb{R}$, or control, $Y(0)\in \mathbb{R}$. Since each unit only receives one treatment and, in turn, only has one observed outcome, causal inference is essentially an identification problem. 

Under the nonparametric setup, the statistical universe $\mathcal{S}$ is the space of all possible joint distribution on $(Z,Y(1),Y(0))$, the observation mapping is $\lambda(S) = (P(Y^* | Z), P(Z))$, where $Y^* = ZY(1) + (1-Z)Y(0)$ is the observed outcome, and the estimand mapping $\theta(S) = E[Y(1) - Y(0)]$ is the average causal effect.

The estimand mapping naturally splits into two parts,
\[
  \theta(S) = \underbrace{E[Y(1)|Z=1]P(Z=1) - E[Y(0)|Z=0]P(Z=0)}_{\theta_a}+ \underbrace{E[Y(1)|Z=0]P(Z=0) - E[Y(0)|Z=1]P(Z=1)}_{\theta_b}.
\]
To determine the identifiability of the first term, write
\[
  \theta_a(S) = \underbrace{E[Y(1)|Z=1]}_{\theta_2(S)} \underbrace{P(Z=1)}_{\theta_1(S)} - \underbrace{E[Y(0)|Z=0]}_{\theta_3(S)}\underbrace{P(Z=0)}_{1-\theta_1(S)}.
\]
By Proposition \ref{proposition:id-function}, both $\theta_1(S)$ and $1 - \theta_1(S)$ are identifiable as they are a simple function of $\lambda(S)$. To see that $\theta_2(S)$ is identifiable, notice that
 \[
 \theta_2(S)  = E[Y(1)|Z=1] = E[ZY(1) + (1-Z)Y(0)|Z=1] = E[Y^* | Z=1]
 \]
 is a function of $\lambda(S)$, Proposition \ref{proposition:id-function} establishes the result. Similarly, $\theta_3(S)$ is identifiable. Applying Proposition \ref{proposition:id-operations} stitches the results showing that $\theta_a$ is identifiable. 

The second term, 
\[
  \theta_b(S) = \underbrace{E[Y(1)|Z=0]}_{\theta_4(S)} \underbrace{P(Z=0)}_{1-\theta_1(S)} - \underbrace{E[Y(0)|Z=1]}_{\theta_5(S)}\underbrace{P(Z=1)}_{\theta_1(S)},
\]
is not identifiable because both $\theta_4(S)$ and $\theta_5(S)$ are generally not. To show $\theta_4(S)$ is not identifiable, assume the outcomes are binary and consider 
\begin{align*}
S_1 &= (P(Y(1), Y(0)) = (1,0) | Z=1)= 1, P((Y(1),Y(0)) = (0,0) | Z =0) =1, P(Z=1) = \frac{1}{2}  ) \ \text{and} \\
 S_2 &= (P(Y(1), Y(0)) = (1,0) | Z=1)= 1, P((Y(1),Y(0)) = (1,0) | Z =0) =1, P(Z=1) = \frac{1}{2}  ).
\end{align*}
Clearly, $\lambda(S_1) = \lambda(S_2)$ while $\theta_4(S_1)=0 \ne 1 = \theta_4(S_2)$, applying Definition \ref{D:identifiability-function} shows that $\theta_4$ is not identifiable. A similarly approach shows that $\theta_5(S)$ is not identifiable. 

To derive the reduced form for the identification region it remains to show that $\theta_b$ is strongly non-identifiable. Fix $\ell_0 = (P^{(0)}(Y^*|Z)P^{(0)}(Z))$, and let $\vartheta_1, \vartheta_2,$ and $\vartheta_3$ be the unique values in $\Theta$ corresponding to $\vartheta_1 R_{\theta_1,\lambda}\ell_0$, $\vartheta_2 R_{\theta_2,\lambda}\ell_0$, and $\vartheta_3 R_{\theta_3,\lambda}\ell_0$, respectively. Let $Y^m = Y(1)(1-Z) + Y(0)Z$ by the unobserved potential outcome. Notice that we can generally write elements of $\Sall$ as $S = (P(Y^m| Y^*,Z) P(Y^*|Z) P(Z)) $. For $\alpha \in \mathbb{R}$, let $S_\alpha = (P^{(\alpha)}(Y^m| Y^*,Z) P^{(0)}(Y^*|Z) P^{(0)}(Z)) $, where $P^{(\alpha)}(Y^m| Y^*,Z)$ is a distribution independent of $Y^*$ with mean $\alpha$. For all values of $\alpha$, $\lambda(S_\alpha) = \ell_0$ and $\theta_b(S_\alpha) = \alpha$; giving us that $H\{\theta_b, \ell_0\} = \mathbb{R}$, which shows that $\theta_b$ is strongly non-identifiable. The reduced form is then, 

\[
H\{\theta, \ell_0\} = \left\{(\vartheta_2\vartheta_1 - \vartheta_3 (1-\vartheta_1) - \vartheta),\  \vartheta \in H\{\theta_b\}  \right\} = \mathbb{R}.
\]

We now consider how the reduced form is impacted by different assumptions. Assume the data will be collected from a Bernoulli randomized experiment, that is $Z$ is independent of $(Y(1),Y(0))$ or, using the notation from the previouse section, $A_1(S) = \delta(P(Y(1),Y(0),Z) ,P(Y(1), Y(0))P(Z))$. This assumption is \emph{a priori} irrefutable as $Img_{A_1}(\lambda) = Img(\lambda)$.

Assumption $A_1$ makes $\theta_4(S)$ point identifiable as $\theta_4(S) = E[Y(1)|Z=0] = E[Y(1)|Z=1] = \theta_2(S)$ --- which we already showed is identifiable. Similarly $\theta_5(S)= \theta_3(S)$, and is also identifiable. Some basic algebra shows that the reduced form for $\ell_0 \in \Lambda$ is the singleton,
\[
H\{\theta, \ell_0\} = \left\{\vartheta_2 - \vartheta_3 \right\}.
\]

Now consider the alternative assumption asserting that $Y(1),Y(0) \in [0,1]$. This assumption is \emph{a priori} refutable as $Img_{A_2}(\lambda) \ne Img(\lambda)$, that is, we have changed the image of the observation mapping. In this setting, it is useful to rewrite the estimand mapping as,
\begin{align*}
  \theta(S) =& E[Y(1)|Z=1]P(Z=1) + E[Y(1)|Z=0]P(Z=0) - E[Y(0)|Z=1]P(Z=1) - E[Y(0)|Z=0]P(Z=0) \\
  =& E[Y^*|Z=1] P(Z=1) - E[Y^*|Z=0](1-P(Z=1)) \\
  \ & + \left(E[Y(1)|Z=0] - E[Y(0)|Z=1] \frac{P(Z=1)}{1-P(Z=1)}\right) (1-P(Z=1))
\end{align*}
The identification region for $\theta$ is,
\begin{align*}
  H^{A_2}\{\theta, \ell_0\}  = \left\{\left(\vartheta_2\vartheta_1 - \vartheta_3 (1-\vartheta_1) + \left[\vartheta - \frac{\vartheta' \vartheta_1}{1-\vartheta_1}\right](1-\vartheta_1)\right),\  \vartheta \in H^{A_2}\{\theta_4\}, \vartheta' \in H^{A_2}\{\theta_5\}  \right\} 
\end{align*}
Since both $\theta_4$ and $\theta_5$ are strongly non-identifiable, the $Img(\theta_4)=Img(\theta_5) = [0,1]$. We can then rewrite the identification region as
\begin{align*}
  H^{A_2}\{\theta, \ell_0\}  &= \left\{\left[\vartheta_2\vartheta_1 - \vartheta_3 (1-\vartheta_1) + (1-\vartheta_1)\right] - \vartheta, \ \vartheta \in [0,1]  \right\} \\
  &= \left\{\left[\vartheta_2\vartheta_1 - \vartheta_3 (1-\vartheta_1) \right] - \vartheta, \ \vartheta \in [-\vartheta_1,1 - \vartheta_1]  \right\}.
\end{align*}

If we further restricted the treatment assignment probabilities to be between 0.4 and 0.6, the reduced form would become
\[
  H^{A_2}\{\theta, \ell_0\}  = \left\{\left[\vartheta_2\vartheta_1 - \vartheta_3 (1-\vartheta_1) \right] - \vartheta, \ \vartheta \in [-0.6,0.6]  \right\}.
\]
The reduced form naturally provides nonparametric bounds for the causal effects. Adding further assumptions can tighten the bound until we are left with a single point, as was the case when we assumed the data were collected from a Bernoulli randomized experiment.

\section{Discussion}\label{section:conclusion}

In this paper, we propose a unifying perspective on identification. Our theory centers around the idea that identifiability can be defined in terms of the injectivity of a certain binary relation. Examining the literature through this lens, we show that existing ad-hoc definitions are special cases of our general framework. One benefit of our flexible formulation is that it brings a new level of transparency and transferability to the concept of identification, allowing us to apply it in settings in which traditional definitions can not be used (Examples~\ref{example:fixed-margins},~\ref{ex:missing-data}, ~\ref{example:ecological-regression}). In addition to providing a flexible --- and completely general --- definition of identifiability, we formalize a three-step process, called identification analysis, for studying identification problems.%
Identification logically precedes estimation: this paper has focused exclusively on the former. A challenge, when thinking about these concepts, is that identification deals with the idealized ``infinite number of observations'' setting, while estimation happens in finite samples. A quantity can, therefore, be identifiable, but difficult to estimate precisely in practice. Nevertheless, thinking about identification first is a critical step, and we hope that our framework will help in that regard.

\bibliography{ref}

\end{document}